\documentclass[review]{elsarticle}
\usepackage{amsmath,graphicx,subfigure,booktabs,colortbl,amsfonts}
\usepackage{amsmath,amssymb,amsthm,algorithm,algpseudocode}
\usepackage{mathrsfs,url,bm}
\biboptions{compress}

\newtheorem{Rem}{Remark}

\newtheorem{proposition}{Proposition}
\begin{document}
\begin{frontmatter}
\title{Statistical Information Fusion for Multiple-View Sensor Data in Multi-Object Tracking}
\author[China]{Xiaoying~Wang}
\author[Australia]{Reza~Hoseinnezhad\corref{mycorrespondingauthor}}
\cortext[mycorrespondingauthor]{Corresponding author}
\ead{rezah@rmit.edu.au}
\author[Australia]{Amirali~K.~Gostar}
\author[Australia]{Tharindu~Rathnayake}
\author[China]{Benlian Xu}
\author[Australia]{Alireza Bab-Hadiashar}
\address[China]{Changshu Institute of Technology, Changshu, Soochow, P. R. of China.}  
\address[Australia]{School of Engineering, RMIT University, Victoria, Australia.}  

\begin{abstract}
This paper presents a novel statistical information fusion method to integrate multiple-view sensor data in multi-object tracking applications. The proposed method overcomes the drawbacks of the commonly used Generalized Covariance Intersection method, which considers constant weights allocated for sensors. Our method is based on enhancing the Generalized Covariance Intersection with adaptive weights that are automatically tuned based on the amount of information carried by the measurements from each sensor. To quantify information content, Cauchy-Schwarz divergence is used. Another distinguished characteristic of our method lies in usage of the Labeled Multi-Bernoulli filter for multi-object tracking, in which the weight of each sensor can be separately adapted for each Bernoulli component of the filter. The results of numerical experiments show that our proposed method can successfully integrate information provided by multiple sensors with different fields of view. In such scenarios, our method significantly outperforms the state of art in terms of inclusion of all existing objects and tracking accuracy.
\end{abstract}

\end{frontmatter}
\section{Introduction}
\label{sec:intro}
Statistical sensor fusion methods combine the information received from multiple sensors to propagate statistical density and estimate the state(s) of object(s) with enhanced accuracy compared to  using a single sensor~\cite{Fusion_book_2009}. The emergence of new sensors, advanced processing techniques, and improved processing hardware has made real-time multi-sensor fusion increasingly viable and rapidly evolving. By means of multi-sensor information fusion, multi-object tracking systems can not only enlarge spatial/temporal surveillance coverage, but also improve system reliability and robustness~\cite{Fusion_book_2004,m_delta_glmb_SPL_2016}. 

Popular statistical multi-object tracking techniques include Multiple-Hypotheses Tracking (MHT)~\cite{MHT_TAC_1979}, Joint Probabilistic Data Association (JPDA)~\cite{MHT_JPDA_book_1999}, and Random Finite Set (RFS) filters~\cite{book_Mahler_2014}. Recent studies on RFS theory have led to various filters with different implementations (based on simplifying but reasonable assumptions on the multi-object distribution) such as probability hypothesis density (PHD) filter and its cardinalized version (CPHD)~\cite{book_Mahler_2014},  the multi-Bernoulli filter (MB) and its cardinality-balanced version (CB-MeMBer)~\cite{vo_cbmember_TSP_2009}, and the newly derived  $\delta$-GLMB(Generalized Labeled Multi-Bernoulli)~\cite{vo_glmb_TSP_2014,vo_glmb_TSP_2015} and its special case, the Labeled Multi-Bernoulli (LMB) filters~\cite{vo_glmb_TSP_2013,vo_lmb_TSP_2014}. A robust version of multi-Bernoulli filter that can handle unknown clutter intensity and detection profile has also been introduced in~\cite{robust_MB_JSTSP}. The most recent family of RFS filters are based on using \textit{labeled random finite set} densities which have been shown to admit conjugacy with the general multiple point measurement set likelihood and to be closed under the Chapman-Kolmogorov equation~\cite{vo_glmb_TSP_2013,vo_lmb_TSP_2014}. 

Various implementations of RFS filters based on Sequential Monte Carlo (SMC) and Gaussian-mixture approximations have been presented. These methods have been utilized to solve different tracking problems in various applications such as track-before-detect visual tracking applications~\cite{vo_tbd_TSP_2010,our_pattern_rec_2012,our_TSP_visual_tracking}, and sensor management in target tracking applications~\cite{amir_OSPA_based_TAES_2015,amir_SPL_2013,amir_SP_2016,hoang_automatica_2014}. 

In multi-sensor applications, a common approach is to use the above filters in each sensor node and fuse the posterior densities that computed locally in every node with each other. The fusion operation can be executed in a distributed or central manner. In this context, multi-sensor fusion refers to combining several multi-object posteriors that depending on the type of the filter used in each sensor node, can have a different mathematical form. To combine the local posteriors, the most common method proposed in the RFS filtering literature is using the Generalized Covariance Intersection (GCI) rule. Examples include the fusion of Poisson multi-object posteriors of multiple local PHD filters~\cite{battistelli_consensus_PHD_SPIE_2015}, i.d.d. clusters densities of several local CPHD filters~\cite{battistelli_consensus_CPHD_2013}, multi-Bernoulli densities of local multi-Bernoulli filters~\cite{bailu_2016}, and LMB or $\delta$-GLMB filters~\cite{consensus_labelled_RFS_2016}. 

The multi-sensor fusion rules that have been formulated based on the GCI-rule in~\cite{battistelli_consensus_PHD_SPIE_2015,battistelli_consensus_CPHD_2013,bailu_2016,consensus_labelled_RFS_2016} all lead to \textit{consensus fusion} (also called \textit{competitive fusion}) of multiple densities. The fused density is formed with more emphasis on objects that are included in \textit{all} local posteriors. In many applications, where different sensors have different fields of view, some objects may not be detected by \textit{all} sensors. In a consensus/competitive sensor fusion scheme, e.g. using the GCI-rule, such objects may not be represented well in the fused multi-object density, and may be lost from the fused multi-object state estimate (i.e. from the tracking results).

Figure~\ref{fig:com_case} shows an example where multiple targets are detected by a multi-transceiver suite. Due to the distance between, and limited field of view of the sensors, each target is detected only by some  sensors. In this scenario, the multi-sensor fusion solution should be capable of combining \textit{complementary} information provided by various sensors. Examples of complementary fusion solutions include the usage of linear-like complementary filters for attitude estimation~\cite{Complementary_TCST_2016}, using an extended Kalman filter for complementary fusion of multi-sensor data in mobile robotics~\cite{Complementary_TM_2015}, and using an unscented Kalman filter for complementary fusion of multiple Poisson densities from local PHD filters in robotic applications~\cite{Complementary_ITSC_2016}.

\begin{figure}
	\centering
	\includegraphics[width=3 in]{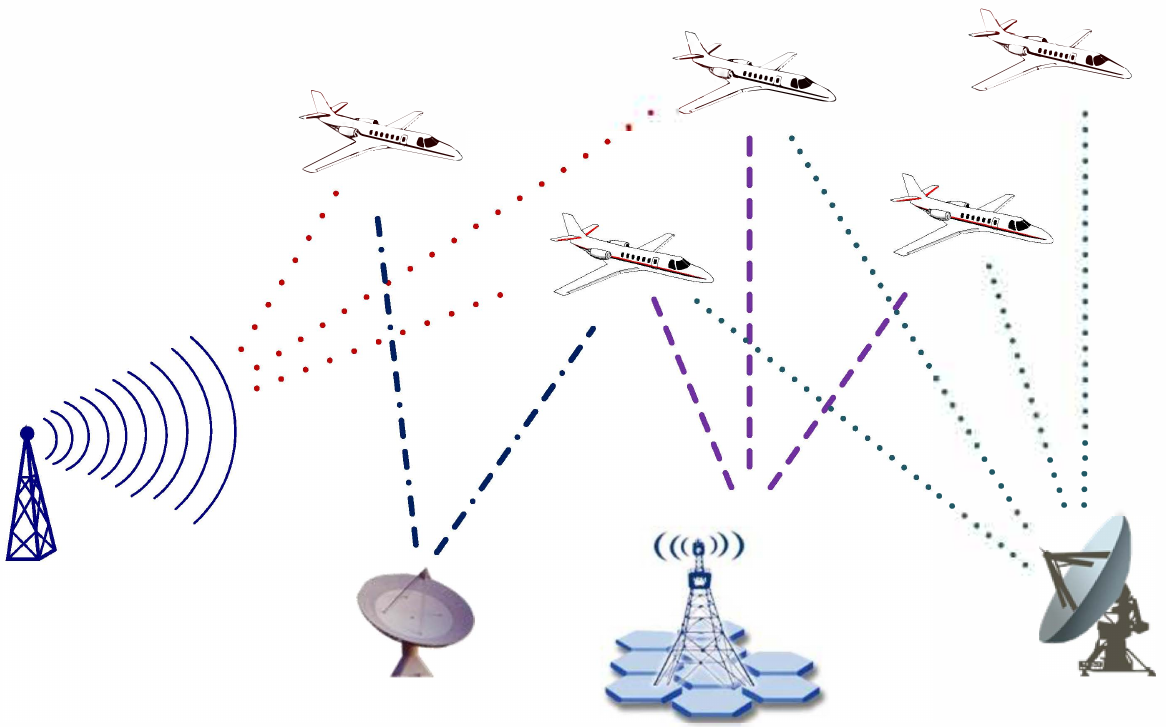}
	\caption{Due to the far distance between the sensors (transceivers) and their limited field of view, no target is consistently detected by \textit{all} sensors. There is no unanimous agreement on any target between all sensors, and a consensus-based fusion algorithm may fail to include all target state estimates.\label{fig:com_case}}
\end{figure}

This paper presents a novel strategy for combining random set posteriors from sensor nodes with different and limited fields of view. The proposed sensor fusion method is fundamentally different from the GCI-rule as used in~\cite{battistelli_consensus_PHD_SPIE_2015,battistelli_consensus_CPHD_2013,bailu_2016,consensus_labelled_RFS_2016}, in the sense that it performs not only consensus fusion (thus, enhanced accuracy when sensors are in agreement on the same object), but also complementary fusion (thus, extended coverage when different sensors cover different regions). 

Our method is particularly useful for fusion of LMB posteriors, with which the GCI-rule leads to separately combining the densities of Bernoulli components (possible objects) with the same label. We give adaptive weights to each Bernoulli component label, so that \textit{the more informative} local posteriors for that label influence the fused density more. The information content of each local posterior can be quantified using  information theoretic divergence. 

In a Bayes filtering context, divergence functionals are commonly used to quantify the expected information gain from prior to posterior densities. The commonly used divergence functions in the stochastic filtering literature include the Shannon entropy~\cite{Shannon_entropy_1991}, the Kullback-Leibler (KL) divergence~\cite{battistelli_consensus_PHD_SPIE_2015,KL_Battistelli_automatica_2014,KL_Battistelli_fusion_2015}, and R\'enyi divergence~\cite{ristic_automatica_2010,ristic_reward_PHD_TAES_2011,amir_SP_2016} which has been the dominant choice of divergence in the random set filtering literature. These divergences have significant computational cost and can only be computed analytically for simple cases. Recently, Cauchy-Schwarz divergence was shown to admit a simple closed-form expression for Poisson multi-object densities~\cite{vo_CS_Poisson}. This was followed by an increasing uptake of Cauchy-Schwarz divergence to solve multi-object tracking problems using labeled random set filters~\cite{vo_CS_fusion2015,vo_void_cs_2015,Amir_CauchySchwarz_2016,meng_CS_2016}. In this work, we use Cauchy-Schwarz divergence to quantify the information content of each local posterior in relation to each possibly existing object (Bernoulli component of the fused LMB posterior).

The major contributions of this paper include:
\begin{enumerate}
\item A complete framework is proposed for statistical multi-sensor fusion in multi-object systems, formulated for LMB filters running locally in each sensor node of the system. The solution can handle both consensus (competitive) fusion of sensors that detect the same object, and complementary fusion of sensors with limited field of view that cover different parts of the object state space, or a mix of competitive and complementary fusion in cases where the sensors' fields of view overlap and objects can be detected by more than one sensor.
\item Formulation of Cauchy Schwarz divergence from prior to posterior is presented for each component of an LMB density, in addition to its incorporation into the fusion formula.
\item Detailed Sequential Monte Carlo (SMC) implementation of the entire framework using LMB filters is presented via a step-by-step algorithm.
\end{enumerate}

The rest of this paper is organized as follows. Section~\ref{sec:back} briefly reviews the necessary background material on labeled multi-Bernoulli filter and Cauchy-Schwarz divergence. The proposed multi-sensor fusion method and its full implementation are then presented in section~\ref{sec:approach}. Numerical experiments are given in section~\ref{sec:sim_res}. They demonstrate that in scenarios involving sensors with unlimited fields of view (hence all objects being detectable, and consensus fusion being suitable), our method performs very similar to the GCI-rule. However, with sensors having limited (and overlapping) fields of view, our method is capable of tracking all the objects while the state of art method (GCI-rule method) leads to targets being lost.  Section~\ref{sec:conc} concludes the paper.
\section{Background}
\label{sec:back}
This section briefly presents background material on Bayesian multi-object tracking, labeled RFS and Cauchy-Schwarz divergence which are necessary for developing our proposed method.
\subsection{Bayesian Multi-Object Tracking}
For notational purposes, we use lower-case letters (e.g. $x,\mathbf{x}$) to denote single-object states and upper-case letters (e.g. $X,\mathbf{X}$) to denote multi-object states. In order to distinguish labeled states and distributions from unlabeled ones, we adopt bold letters to represent labeled entities and variables (e.g. $\mathbf{X,x,\bm{\pi}}$). Furthermore, blackboard bold letters represent spaces for variables (e.g. $\mathbb{X,Z,L}$).

Assume that at time $k$, the labeled multi-object state is denoted by $\textbf{X}_k\subset\mathbb{X}$ and the multi-object observation is denoted by $Z_k\subset\mathbb{Z}$ where $\mathbb{X}$ and $\mathbb{Z}$ denote the space of single-object states and single-object measurements, respectively. Both $\textbf{X}_k$ and $Z_k$ are modeled as random finite sets. In a Bayes multi-object filtering scheme, the multi-object random set distribution is recursively predicted and updated.

Let us denote the labeled multi-object prior density (at time $k-1$) by $\bm{\pi}_{k-1}(\cdot|Z_{1:k-1})$, where $Z_{1:k-1}$ is the collection of finite measurements up to time $k-1$. The multi-object prediction at time $k$ is given by Chapman-Kolmogorov equation,
\begin{equation}
\bm{\pi}_{k|k-1}(\bm{X}_k|Z_{1:k-1})=\hspace{-2mm}\int f_{k|k-1}(\bm{X}_k|\bm{X})\bm{\pi}_{k-1}(\bm{X}|Z_{1:k-1})\delta \bm{X},
\end{equation}
where $ f_{k|k-1}(\bm{X}_k|\bm{X})$ is the multi-object transition kernel from time $k-1$ to time $k$ and the integral is the labeled set integral defined in~\cite{vo_lmb_TSP_2014,vo_CS_fusion2015}. 

The update step of the Bayes multi-object filter at time $k$ is based on Bayes' rule and returns the following multi-object posterior:
\begin{equation}
\bm{\pi}_{k}(\bm{X}_k|Z_{1:k})=\frac{g_k(Z_k|\bm{X}_k)\bm{\pi}_{k|k-1}(\bm{X}_k|Z_{1:k-1)}}
{\int g_k(Z_k|\bm{X})\bm{\pi}_{k|k-1}(\bm{X}|Z_{1:k-1)}\delta \bm{X}},
\end{equation}
where $Z_k$ is the set of observations received at time $k$, and it is comprised of some measurements each associated with an object (with some objects possibly missed), and some false alarms or clutter. Both the number of object-related measurements and the number of false alarms randomly vary with time. Hence, $Z_k$ is a random set with its stochastic variations characterized by a multi-object likelihood function $g_k(Z_k|\bm{X}_k)$.

\subsection{Labeled Multi-Bernoulli Filter}

To implement the general Bayes' filtering steps presented in previous section, for the sake of computational tractability, a specific family of distribution is assumed for the multi-object state, and the filter is implemented to predict and update the parameters of that specific multi-object distribution. Labeled multi-Bernoulli (LMB) is an example of the recently introduced labeled random set distributions. The Bayes' filter that operates based on assuming LMB prior and posterior is called an LMB filter. 

The LMB distribution is completely described by its components $\bm{\pi} = \{(r^{(\ell)},p^{(\ell)}(\cdot))\}_{\ell\in\mathbb{L}}$ where $r^{(\ell)}$ is the \textit{probability of existence} of a track with label $\ell \in \mathbb{L}$, and $p^{(\ell)}(x)$ is the probability density of the object's state $x \in \mathbb{X}$ conditional on its existence. The LMB RFS density is given by
\begin{equation}
\label{eq:LMB_Distribution}
\bm{\pi}(\mathbf{X})=\Delta (\mathbf{X})w (\mathcal{L}(\mathbf{X}))\left[ p\right] ^{\mathbf{X}},  
\end{equation}
where $\mathcal{L}(\bm{X})$ is the set of all labels extracted from labeled states in $\bm{X}$, and 
\begin{equation}
\Delta(\mathbf{X}) \triangleq
\begin{cases}
1 & \mathrm{if}\ |\mathbf{X}| = |\mathcal{L}(\mathbf{X})| \\
0 & \mathrm{otherwise,}
\end{cases}
\end{equation}
 in which $|\cdot|$ means ``the cardinality of", and 
 \begin{equation}
 [p]^{\bm{X}} \triangleq \prod_{(x,\ell)\in\bm{X}} p^{(\ell)}(x),
 \end{equation}
 and
\begin{eqnarray}
w(L) &=& \prod\limits_{i \in \mathbb{L}}\left( 1-r^{(i)}\right) \prod\limits_{\ell \in L} \frac{1_{\mathbb{L}}(\ell)r^{(\ell)}} {(1-r^{(\ell)})}
\end{eqnarray}
is the probability of joint existence of all objects with labels $\ell \in L$ and non-existence of all other labels~\cite{vo_lmb_TSP_2014}. 

In the LMB filter, suppose that the prior is an LMB with parameters $ \{(r^{(\ell)},p^{(\ell)}(\cdot))\}_{\ell\in\mathbb{L}}$. In the prediction step, the LMB prior is turned into the following new LMB density with evolved particles and probabilities of existence including the LMB birth components:
\begin{equation}
{\bm{\pi}}_{+} = \left\{\left(r_{+,S}^{(\ell)},p_{+,S}^{(\ell)}\right)\right\}_{\ell\in\mathbb{L}} \cup 
\left\{\left(r_{B}^{(\ell)},p_{B}^{(\ell)}\right)\right\}_{\ell\in\mathbb{B}}
\label{eq:prediction}
\end{equation} 
where 
\begin{eqnarray}
r_{+,S}^{(\ell)} & = & \eta_S(\ell)\, r^{(\ell)} \\
p_{+,S}^{(\ell)} & = & \langle p_S(\cdot,\ell) f(x|\cdot,\ell),p^{(\ell)}(\cdot) \rangle / \eta_S(\ell)
\end{eqnarray}
and 
\begin{equation}
\eta_S(\ell) = \langle p_S(\cdot,\ell),p^{(\ell)}(\cdot) \rangle.
\label{eq:prediction2}
\end{equation}
We denote the predicted LMB parameters by $\{(r_+^{(\ell)},p_+^{(\ell)}(\cdot))\}_{\ell\in\mathbb{L}_+}$ where $\mathbb{L}_+ = \mathbb{L}\cup\mathbb{B}$. 

Assume that at the sensor node $i$, a measurement set denoted by $Z_i$ is acquired, and denote the LMB posterior (locally updated at sensor node $i$) as  
$$
{\bm{\pi}}_{i}(\cdot|Z_i) = \left\{
\left(
r_{i}^{(\ell)},p_{i}^{(\ell)}(\cdot)
\right)
\right\}_{\ell\in\mathbb{L}_+}.
$$
The parameters of each local LMB posterior are given by~\cite{vo_lmb_TSP_2014}:
\begin{equation}
\begin{array}{rcl}
r_{i}^{(\ell)} & = & \sum\limits_{(I_+,\theta)\in\mathcal{F}(\mathbb{L}_+)\times{\Theta}_{I_+}}
w^{(I_+,\theta)}(Z)\ 1_{I_+}(\ell)\\
p_{i}^{(\ell)}(x) & = & \frac{1}{r^{(\ell)}} \sum\limits_{\scriptsize{(I_+,\theta)\in\mathcal{F}(\mathbb{L}_+)\times{\Theta}_{I_+}}}
\hspace{-3mm}w^{(I_+,\theta)}(Z)\ 1_{I_+}(\ell) p^{(\theta)}(x,\ell)
\end{array}
\label{eq:update}
\end{equation}
where
\begin{equation}
\begin{array}{rcl}
w^{(I_+,\theta)}(Z) & \propto & w_{+}(I_+) [\eta_Z^{(\theta)}]^{I_+} \\
p^{(\theta)}(x,\ell) & = & \frac{p_+^{(\ell)}(x) \psi_Z(x,\ell;\theta)}{\eta_Z^{(\theta)}(\ell)} \\
\eta_Z^{(\theta)}(\ell) & = & \langle  p_+^{(\ell)}(x),\psi_Z(x,\ell;\theta) \rangle \\
\psi_Z(x,\ell;\theta) & = & 
\left\{
\begin{array}{lcr}
\frac{p_D(x,\ell) g(z_{\theta(\ell)}|x,\ell)}{\kappa(z_{\theta(\ell)})}, & \mathrm{if} & \theta(\ell)>0 \\
1-p_D(x,\ell), & \mathrm{if} & \theta(\ell) = 0
\end{array}
\right.
\end{array}
\end{equation}
and $\Theta_{I_+}$ is the space of mappings $\theta: I_+ \rightarrow \{0,1,\ldots,|Z|\}$ such that $\theta(i) = \theta(i')>0$ implies $i = i'$,
and the weight term, $w_{+}(I_+)$, is given by:
\begin{equation}
w_{+}(I_+) = \prod_{i\in \mathbb{L}_+} \left(1-r_+^{(i)}\right) 
\prod_{\ell\in I_+}
\frac{1_{\mathbb{L}_+}(\ell) r_+^{(\ell)}}{1-r_+^{(\ell)}}.
\label{eq:update2}
\end{equation}

When the filter is implemented using SMC method, the density function of each component with label $(\ell)$ is approximated by $J^{(\ell)}$ number of particles and weights 
\begin{equation}
p^{(\ell)}(x) \approx \sum_{j=1}^{J^{(\ell)}} w_{j}^{(\ell)}
\delta(x-x_{j}^{(\ell)})
\end{equation}
where $\delta(\cdot)$ is the Dirac delta function. This makes it easy to compute the inner products and integrals involved in prediction and update equations (integrals are replaced with sums over particles). 

\subsection{ {C}auchy-{S}chwarz divergence}
The Cauchy-Schwarz divergence is based on the Cauchy-Schwarz inequality for inner products. The Cauchy-Schwarz divergence between the probability densities $f$ and $g$ of two random vectors with respect to the reference measure $\mu$ is given by:
\begin{equation}
D_{CS}(f,g)= - \ln\frac{\langle f, g \rangle_\mu}{\left\| f \right\|_\mu\left\| g \right\|_\mu}
\end{equation}
where $\langle f , g \rangle_\mu  \triangleq \int_{\mathbb{X}} f(x) g(x)\mu dx$
denotes standard inner product of two densities, and $\left\| f \right\|_\mu\triangleq\sqrt{\langle f, f \rangle_\mu}$. Note that $D_{CS}(f,g)$ is symmetric and positive unless $f = g$, in which case $D_{CS}(f,g)=0$.  

Similar to random vector densities, the Cauchy-Schwarz divergence between two densities $\bm{\pi}_1$ and $\bm{\pi}_2$ is given by:
\begin{equation}
D_{CS}(\bm{\pi}_1,\bm{\pi}_2)= - \ln\frac{\langle \bm{\pi}_1, \bm{\pi}_2 \rangle_\mu}{\left\| \bm{\pi}_1 \right\|_\mu\left\| \bm{\pi}_2 \right\|_\mu}
\end{equation}
where the norms are computed via set integrals~\cite{vo_CS_Poisson}. The computational cost of the above divergence term is usually too expensive for real-time applications. However, as Hoang~\textit{et al.}~\cite{vo_CS_Poisson} showed, in the special case of two Poisson RFSs, the Cauchy-Schwarz divergence equals half the squared distance between their intensity functions. Indeed, denoting the two Poisson RFS densities by $\pi_1$ and $\pi_2$ and their intensity functions by  $v_1(\cdot)$ and $v_2(\cdot)$, the Cauchy-Schwarz divergence between them is given by: 
\begin{equation}
D_{CS}(\pi_1,\pi_2)= \frac{K}{2}\|v_1-v_2\|^2
\label{eq:CSD_PHD}
\end{equation}
where $\|v_1-v_2\|^2 = \int |v_1(x)-v_2(x)|^2 dx,$ and $K$ is the unit of hyper-volume.
\section{Multiple-View Sensor Fusion with LMB Filter}
\label{sec:approach}
Consider an LMB filter that operates on the measurements received from $n_s$ sensors for tracking an unknown number of objects. The predicted LMB is independent from sensor measurements and can be directly computed by~\eqref{eq:prediction}--\eqref{eq:prediction2}. In each sensor node $i$, a local update returns a posterior from the prior (predicted) LMB, denoted by $\{(r_{i}^{(\ell)},p_{i}^{(\ell)}(\cdot))\}_{\ell\in\mathbb{L}_+}$. The sensor fusion problem is focused on how to combine all these local posteriors into a fused LMB posterior that conveys all the information in a concise yet comprehensive manner. 

The most common method for sensor fusion with random set filters is the Generalized Covariance Intersection (GCI) method. Let us assume that local posteriors $\bm{\pi}_1, \ldots, \bm{\pi}_{n_s}$ are fused into the multi-object density $\bm{\pi}$. In GCI method, the following weighted average of Kullback-Leibler distances from local posteriors to the fused posterior should be minimized~\cite{consensus_labelled_RFS_2016}:
\begin{equation}
\bm{\pi}_{\mathrm{fused}} = \arg\min_{\bm{\pi}} \sum\limits_{i=1}^{n_s} \omega_i D_{KL}(\bm{\pi}_i,\bm{\pi}). 
\end{equation}
where $\sum_{i=1}^{n_s} \omega_i = 1$. The weights $\{\omega_i\}_{i=1}^{n_s}$ are constant and quantify our confidence on validity of the measurements acquired from each sensor. 

When all the local and fused posteriors are LMB, with the above mentioned parametrization, the following GCI-rule is derived~\cite{consensus_labelled_RFS_2016,meng_CS_2016}:
\begin{flalign}
\tilde{r}^{(\ell)} = & \frac{\int \prod_{i=1}^{n_s} \left(r_{i}^{(\ell)} p_{i}^{(\ell)}(x)\right)^{\omega_i}dx}{
	\prod_{i=1}^{n_s}\left(1-r_{i}^{(\ell)}\right)^{\omega_i}
	+\int \prod_{i=1}^{n_s} \left(r_{i}^{(\ell)} p_{i}^{(\ell)}(x)\right)^{\omega_i}dx} 
\label{eq:fused_r_1}
\\
\tilde{p}^{(\ell)}(x) = & \frac{
	\prod_{i=1}^{n_s} \left(p_{i}^{(\ell)}(x)\right)^{\omega_i}
}
{\int \prod_{i=1}^{n_s} \left(p_{i}^{(\ell)}(x)\right)^{\omega_i} dx}.
\label{eq:fused_p_1}
\end{flalign}
where $\omega_i$ is a weight indicating the strength of the emphasis on sensor $s_i$ in the fusion process. These weights should be normalized, i.e. $\sum_{i=1}^{n_s} \omega_i = 1$.

In an SMC implementation, let us denote the predicted LMB parameters by $\{(r_+^{(\ell)},\{w_{+j}^{(\ell)},x_{+j}^{(\ell)}\}_{j = 1}^{J_+^{(\ell)}})\}_{\ell\in\mathbb{L}_+}$ and the local LMB posterior at node $i$ by $\{(r_i^{(\ell)},\{w_{ij}^{(\ell)},x_{ij}^{(\ell)}\}_{j = 1}^{J_i^{(\ell)}})\}_{\ell\in\mathbb{L}_+}$. During the update step of each local LMB filter, the particles do not change, and only their weights evolve. Thus, for each label $\ell$ we have,
$$ J_i^{(\ell)} = J_+^{(\ell)};\ \ \ \ x_{i,j}^{(\ell)} = x_{+j}^{\ell}.
$$
In other words, the predicted LMB prior and all the updated LMB posteriors will have the same particles but with different weights and existence probabilities. This makes the fusion of the posteriors generated at each sensor straightforward. 

Substituting each density with its particle approximation turns the integrals into weighted sums over the particles, and the fused existence probability is  given by:
\begin{equation}
\tilde{r}^{(\ell)} = \frac{
	\sum_{j=1}^{J_+^{(\ell)}} \prod_{i=1}^{n_s} \left(r_{i}^{(\ell)} w_{i}^{(\ell)}\right)^{\omega_i}}
{
	\prod_{i=1}^{n_s}\left(1-r_{i}^{(\ell)}\right)^{\omega_i}
	+\sum_{j=1}^{J_+^{(\ell)}} \prod_{i=1}^{n_s} \left(r_{i}^{(\ell)} w_{i}^{(\ell)}\right)^{\omega_i}}.
\label{eq:fused_r_2}
\end{equation}
The fused densities also take the form of weighted sums of Dirac deltas:
\begin{equation}
\tilde{p}^{(\ell)}(x) = \sum_{j=1}^{J_+^{(\ell)}} \tilde{w}_{j}^{(\ell)}\ \delta(x-x_{+j}^{(\ell)})
\label{eq:fused_p_2}
\end{equation}
where 
\begin{equation}
\tilde{w}_{j}^{(\ell)} = 
\frac{
	\prod_{i=1}^{n_s} \left(w_{i,j}^{(\ell)}\right)^{\omega_i}
}
{ 
	\sum_{j=1}^{J_+^{(\ell)}}
	\prod_{i=1}^{n_s} \left(w_{i,j}^{(\ell)}\right)^{\omega_i}
}
\label{eq:weight_GCI_fusion}
\end{equation}
is the fused weight of each fused posterior particle.

\begin{Rem}
By minimizing the average distance from \textit{all} local posteriors, the GCI fusion method is considered a \textit{consensus-based} fusion method. The fused posterior should be as similar as possible to \textit{all} local posteriors, hence, would convey the information that are overlapped (regions of consensus) between all sensors. This is correct regardless of what the constant weights $\{\omega_i\}_{i=1}^{n_s}$ are as long as they are non-zero.
\label{rem:GCI_is_consensus}
\end{Rem}

\begin{Rem}
From equations~\eqref{eq:fused_r_2} and~\eqref{eq:weight_GCI_fusion}, it is evident that during the GCI fusion, parameters of each Bernoulli component with label $\ell$ are only fused with those with the same label. Indeed, the fusion could be implemented separately for each label, and there is no cross-label computing.
\label{rem:label_separation}
\end{Rem}

Based on the observation made in Remark~\ref{rem:GCI_is_consensus}, GCI fusion in its current form is consensus fusion. As it was explained earlier~(see Figure~\ref{fig:com_case} and the associated explanation in section~\ref{sec:intro}), consensus fusion cannot be successfully applied to combine the measurements acquired from sensors with different fields of view. In the following, we present a novel solution that resolves this issue.

Inspired by the observation that in GCI fusion, for each label, the multi-sensor data are fused separately (as stated in Remark~\ref{rem:label_separation}), we suggest that each weight $\omega_i$ associated with each sensor $s_i$ is adaptively tuned separately for each label $\ell$. For this purpose, we propose that for each label $\ell$, the weight associated with sensor $s_i$, denoted by $\omega_i^{(\ell)}$, should be continuously tuned to exponentially increase with the information content of that sensor relevant to the possible object with label $\ell$. To quantify the \textit{information content} of a sensor measurement, we look at the Cauchy Schwarz divergence between its local posterior and the prior, with only label $\ell$ considered.

Suppose the predicted LMB distribution is $\bm{\pi}_+=\{(r_+^{(\ell)},p_+^{(\ell)}(\cdot))\}_{\ell\in\mathbb{L}}$, and at Sensor $S_i$ the measurement set $Z_i$ is acquired and used in a local LMB update leading to the local posterior $\bm{\pi}_i=\{(r_i^{(\ell)},p_i^{(\ell)}(\cdot))\}_{\ell\in\mathbb{L}}$. Associate with a particular label $\ell$, there is a prior Bernoulli RFS with its density parametrized as ${\pi}_+^{(\ell)} = \{(r_+^{(\ell)},p_+^{(\ell)}(\cdot))\}$ and a local posterior Bernoulli RFS with its density parametrized as ${\pi}_i^{(\ell)} = \{(r_i^{(\ell)},p_i^{(\ell)}(\cdot))\}$. The fusion weight $\omega_i^{(\ell)}$ is then given by:
\begin{equation}
{\omega_{i}^{(\ell)}}\propto\exp\frac{D_i^{(\ell)}}{\min_{i'\in [1,n_s]} D_{i'}^{(\ell)}}
\label{eq:CS_weight}
\end{equation}
where $D_i^{(\ell)}$ is the Cauchy-Schwarz divergence between $\pi_+^{(\ell)}$ and $\pi_i^{(\ell)}$. The following proposition presents how the divergence between two Bernoulli RFSs can be computed.

\begin{proposition}
A first moment approximation of the Cauchy-Schwarz divergence between two Bernoulli random set densities $\pi_1$ and $\pi_2$, parametrized by $\{(r_1,p_1(\cdot))\}$ and $\{(r_2,p_2(\cdot))\}$ respectively, is given by:
\begin{equation}
D_{\mathrm{CS}}(\pi_1,\pi_2) \approxeq \frac{1}{2} K \left\| r_1\, p_1(\cdot)-r_2\, p_2(\cdot)\right\|^2
\label{eq:BCSD}
\end{equation}
where $\|v_1(\cdot)-v_2(\cdot)\|^2 = \int |v_1(x)-v_2(x)|^2 dx,$ and $K$ is the unit of hyper-volume.
\label{prop:cs}
\end{proposition}

\begin{proof}
Mahler~\cite{book_Mahler_2014} has shown that in terms of Kullback-Leibler distance, the closest Poisson density $\pi_{\mathrm{Poiss.}}$ to any RFS density $\pi$ is the one that matches the first moment (intensity function) of $\pi$. On the other hand, the intensity function of a Bernoulli RFS with parameters $\{(r,p(\cdot))\}$ is given by:~\cite{book_Mahler_2014}
\begin{equation}
v(x) = r\, p(x).
\end{equation}
Hence, the best first moment approximation of the Bernoulli densities $\pi_1$ and $\pi_2$ are Poisson densities with intensity functions
\begin{equation}
v_1(x) = r_1\, p_1(x);\ \ \ \ \ v_2(x) = r_2\, p_2(x).
\end{equation}
Approximating the two Bernoulli RFSs with their first moment approximations (Poisson RFSs), the Cauchy-Schwarz divergence between them is then given by~\eqref{eq:CSD_PHD}. Substituting the above intensity functions in~\eqref{eq:CSD_PHD} leads to equation~\eqref{eq:BCSD}.
\end{proof}

When the LMB filter is implemented using particle approximations for Bernoulli densities, the inner product integration turns into a weighted sum over particles. Thus, in the context of weight tuning for multi-sensor fusion, the Cauchy-Schwarz divergence between the Bernoulli prior ${\pi}_+^{(\ell)} = \{(r_+^{(\ell)},\{w_{+,j}^{(\ell)},x_{+,j}^{(\ell)}\}_{j=1}^{J_+^{(\ell)}})\}$ and a local Bernoulli posterior ${\pi}_i^{(\ell)} = \{(r_i^{(\ell)},\{w_{i,j}^{(\ell)},x_{+,j}^{(\ell)}\}_{j=1}^{J_+^{(\ell)}})\}$
is given by:\footnote{Note that the total number of particles $J_+^{(\ell)}$ and their locations $x_{+,j}^{(\ell)}$ are the same in the prior and the local posterior, only particle weights and existence probabilities are updated.}
\begin{equation}
D_i^{(\ell)} = \frac{1}{2}K\ \sum_{j=1}^{J_+^{(\ell)}}
\left|
r_i^{(\ell)} w_{i,j}^{(\ell)} - r_+^{(\ell)} w_{+,j}^{(\ell)} 
\right|^2.
\label{eq:CS}
\end{equation}
 
In a step-by-step algorithm, for each possibly existing object with label $\ell$, we first compute the Cauchy-Schwarz divergences $D_i^{(\ell)}$ for all sensor nodes $i \in [1,n_s]$. We then find the minimum divergence $D_{\min}^{(\ell)} \triangleq \min_{i\in[1,n_s]} D_i^{(\ell)}$ and compute the non-normalized weights as 
$$
\Omega_i^{(\ell)} = \exp
\left({D_i^{(\ell)}}\bigg/{D_{\min}^{(\ell)}}\right)
$$ 
and finally normalize the weights as $$\omega_i^{(\ell)} =\frac{\Omega_i^{(\ell)}}{\sum_{i'=1}^{n_s} \Omega_i^{(\ell)}}.$$
With these label-dependent weights that continuously evolve with time (depending on the worth of sensor data relevant to each possible object label), the local posteriors are then fused to form a new LMB with the following probabilities of existence:
\begin{equation}
\tilde{r}^{(\ell)} = \frac{
	\sum_{j=1}^{J_+^{(\ell)}} \prod_{i=1}^{n_s} \left(r_{i}^{(\ell)} w_{i}^{(\ell)}\right)^{\omega_i^{(\ell)}}}
{
	\prod_{i=1}^{n_s}\left(1-r_{i}^{(\ell)}\right)^{\omega_i^{(\ell)}}
	+\sum_{j=1}^{J_+^{(\ell)}} \prod_{i=1}^{n_s} \left(r_{i}^{(\ell)} w_{i}^{(\ell)}\right)^{\omega_i^{(\ell)}}}
\label{eq:fused_r_3}
\end{equation}
and the following weights for the particles in the fused Bernoulli component densities:
\begin{equation}
\tilde{w}_{j}^{(\ell)} = 
\frac{
	\prod_{i=1}^{n_s} \left(w_{i,j}^{(\ell)}\right)^{\omega_i^{(\ell)}}
}
{ 
	\sum_{j=1}^{J_+^{(\ell)}}
	\prod_{i=1}^{n_s} \left(w_{i,j}^{(\ell)}\right)^{\omega_i^{(\ell)}}
}.
\label{eq:weight_GCI_fusion_3}
\end{equation}

\begin{Rem}
To maintain tractability, LMB components with extremely small existence probabilities should be pruned. This is performed only after GCI-fusion of the local posteriors.
\end{Rem}
\subsection{Step-by-step Algorithm}
Algorithm~\ref{alg:1} shows the complete pseudo code for our proposed approach to multiple-view multi-sensor fusion with the LMB filter. Starting with an LMB prior (which is the fused LMB posterior from previous time), the function $\textsc{Predict}(\cdot)$ implements the LMB prediction step and the function $\textsc{Update}(\cdot)$ implements each sensor's local LMB update step.  $\textsc{Fusion}(\cdot)$ denotes the multi-sensor fusion process which is separately detailed in algorithm~\ref{alg:2}. There, the parameters of locally updated posteriors for each label are separately fused using GCI-fusion formulae \eqref{eq:fused_r_3}-\eqref{eq:weight_GCI_fusion_3}, in which the weights of sensors have been adaptively tuned according to the information contents of each sensor measurement relevant to the object label.

In Algorithm~\ref{alg:1}, after the fused LMB posterior $\tilde{\bm{\pi}}$ is acquired, Bernoulli components with small existence probabilities are pruned and finally the number of objects and their states are estimated. 

\begin{algorithm}[!htb]
	\caption{LMB filtering with multiple-view multi-sensor fusion}
	\label{alg:1} 
	\begin{algorithmic}[1] 
	\Procedure{LMB\_Filter}{LMB prior $\bm{\pi}$}
		\State $\bm{\pi}_{+} \gets$ \Call{Predict}{$\bm{\pi}$} \Comment Using Eq. \eqref{eq:prediction}-\eqref{eq:prediction2}.
		\For{$i=1:n_s$}
		\State $\bm{\pi}_{i} \gets$ \Call{Update}{$\bm{\pi}_{+},Z_i$} \Comment Using Eq. \eqref{eq:update}-\eqref{eq:update2}.
		\EndFor
		\State $\tilde{\bm{\pi}} \gets$ \Call{Fusion}{$\{\bm{\pi}_{i}\}_{i=1}^{n_s},\bm{\pi}_+$}
		\State \Call{Pruning}{}	 \Comment Ref. \cite{vo_glmb_TSP_2014,vo_lmb_TSP_2014}
		\State \Call{Estimatation}{} \Comment Ref. \cite{vo_glmb_TSP_2014,vo_lmb_TSP_2014}
	\EndProcedure
	\end{algorithmic}
\end{algorithm}

\begin{algorithm}[!htb]
	\caption{Multi-sensor fusion over separate object labels.}
	\label{alg:2} 
	\begin{algorithmic}[1] 
	\Procedure{Fusion}{$\{\bm{\pi}_{i}\}_{i=1}^{n_s},\bm{\pi}_+$}
	
	\For{$\ell\in \mathbb{L}_+$}
		\For{$i=1:n_s$}
		\State $D^{(\ell)}_i \gets $\Call{CSD}{$\bm{\pi}_{i},\bm{\pi}_+$}  \Comment Using Eq. \eqref{eq:CS}
		\EndFor
		\State $D^{(\ell)}_{\min} \gets \min_{i\in[1,n_s]} D^{(\ell)}_i$
		\For{$i=1:n_s$}
		\State $\Omega^{(\ell)}_i\gets \exp
		\left({D^{(\ell)}_i}\big/{D^{(\ell)}_{\min}}\right)
		$ 
		\EndFor
		\For{$i=1:n_s$}
		\State $\omega^{(\ell)}_i\gets {\Omega^{(\ell)}_i}\big/{\sum_{i'=1}^{n_s}\Omega^{(\ell)}_{i'}}
		$ \Comment Normalizing
		\EndFor		
		\State $\tilde{\bm{\pi}}^{(\ell)} \gets$\Call{GCI-Fusion}{$\{\bm{\pi}_{i}^{(\ell)},\omega^{(\ell)}_i \}_{i=1}^{n_s}$}
		\Statex \Comment Using Eq. \eqref{eq:fused_r_3}-\eqref{eq:weight_GCI_fusion_3}
		\EndFor
	\EndProcedure
	\end{algorithmic}
\end{algorithm}

\section{Numerical Experiments}
\label{sec:sim_res}
We conducted extensive experiments involving various scenarios with different numbers of targets, sensors, target motion models, sensor detection profile models and sensors field of view. In each experiment, we compared the performance of the proposed sensor fusion method with the state of art method, in terms of Optimal Sub-Pattern Assignment (OSPA) errors~\cite{vo_OSPA_TSP_2008}. This section includes the representative set of our simulation results which show how the proposed method works in different clutter rates and sensors' fields of view. All scenarios share the following parameters.

The targets maneuver in an area of $1600~\textrm{m}~\times~1200~\textrm{m}$. The single target state $\bm{x}$ is comprised of its label and unlabeled state. The label is formed as $\ell = (k_B,i_B)$ where $k_B$ is the birth time of the target and $i_B$ is an index to distinguish targets born at the same time. The unlabeled state is four-dimensional and includes the Cartesian coordinates of the target and its speed in those directions, denoted by $x = [p_x\ \dot{p}_x\ p_y\ \dot{p}_y]^\top$. Each target moves according to the Nearly-Constant Velocity (NCV) model. With this model, the transition density is $f_{k|k-1}(x_k|x_{k-1},\ell)=\mathcal{N}(x_k;F_kx_{k-1},Q_k)$, where the transition matrix is
\begin{equation*}
F_k=
\begin{bmatrix}
1 & T & 0 & 0\\
0 & 1 & 0 & 0\\
0 & 0 & 1 & T\\
0 & 0 & 0 & 1
\end{bmatrix},
\end{equation*}
and the covariance matrix is
\begin{equation*}
Q_k=\sigma_w^2
\begin{bmatrix}
B & \bm{0}_2 \\
\bm{0}_2 & B
\end{bmatrix}\ \ ;\ \ B = 
\begin{bmatrix}
0.1 & 0.001 \\ 0.1 & 0.001
\end{bmatrix}.
\end{equation*}
The variance parameter $\sigma_w=5$~m$/$s$^2$
and $T$ is the sampling interval~(in our experiments $T=1\,$s). The probability of survival is fixed at $p_S(x,\ell) = 0.99$ and the LMB component's pruning threshold is $\gamma_p=10^{-4}$. For each sensor $s_i$, with its location denoted by $[x_{s_i}\ y_{s_i}]^\top$, each target (if detected) leads to a bearing and range measurement vector with the measurement noise density given by $ \mathcal{N}(\cdot; [0\ 0]^\top, R)$ in which $R=$~diag$(\sigma_\theta^2,\sigma_r^2)$ with $\sigma_\theta^2=2\pi/180$~rad and $\sigma_r^2=10$~m$^2$ being the scales of range and bearing noise. Thus, the single target likelihood function is ~$g(z_i|x,\ell)=\mathcal{N}(z;\mu_i(x),R)$, where 
\begin{equation}
\mu_i(x) = \left[\arctan\left({\frac{p_x\!-\!x_{s_i}}{p_y\!-\!y_{s_i}}}\right)\ \sqrt{({p_x\!-\!x_{s_i}})^2+({p_y\!-\!y_{s_i}})^2}\right]^\top.
\end{equation}

 In all scenarios, the density $p^{(\ell)}(\cdot)$ of each labeled Bernoulli component in the filter is approximated by $J^{(\ell)} = 1000$ particles. All simulation experiments were implemented in MATLAB R2015b and ran on an Intel® Core™ i7-4770 CPU @3.40GHz, and 8~GB memory. 

\begin{figure}
\centering
\includegraphics[width=3 in]{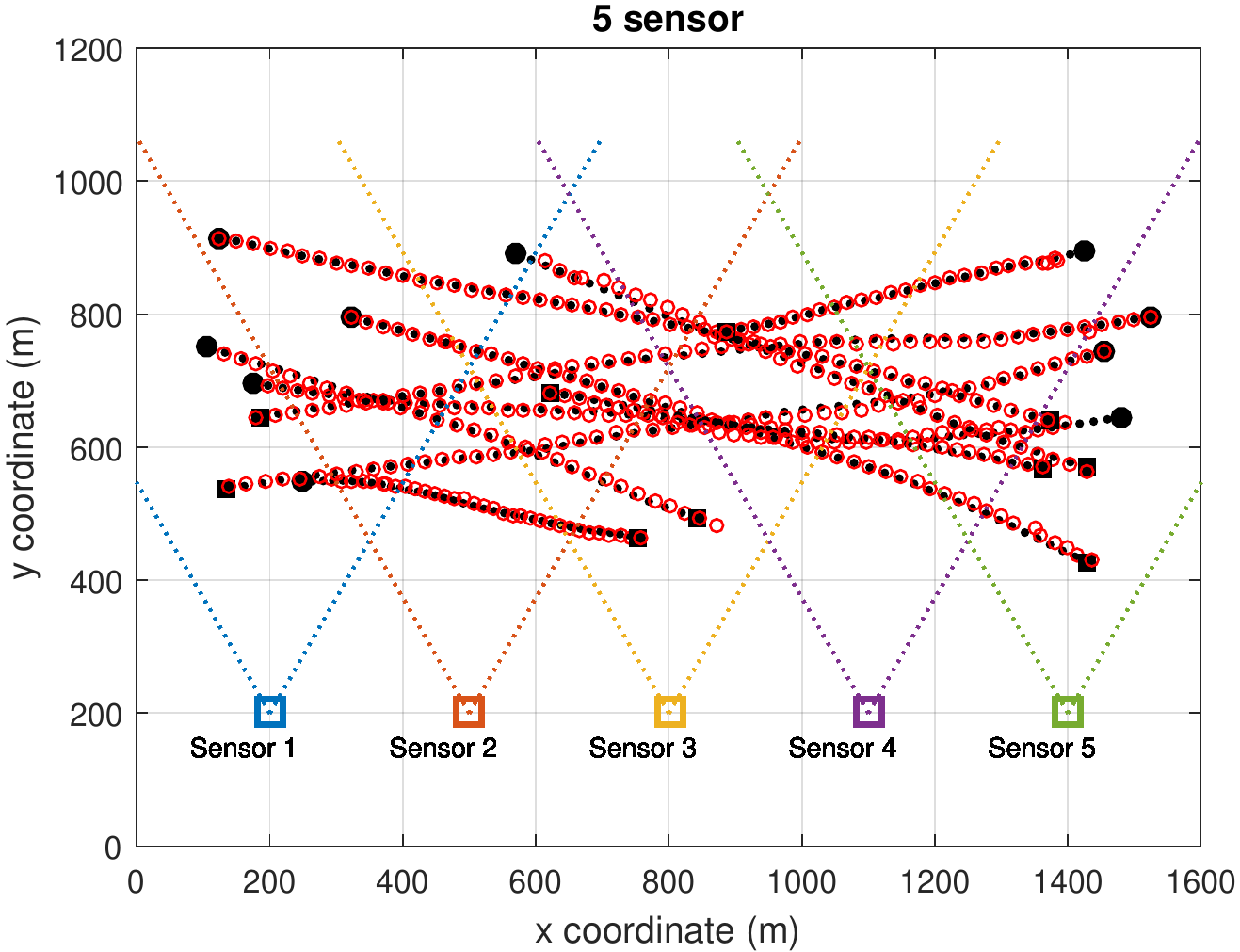}\\
\caption{\label{fig:scenario_10target} A challenging simulated scenario involving 10 targets appearing and disappearing in different locations and times. Multi-target tracking needs fusion of the multiple sensors that have different fields of view. Black dots are ground truths targets trajectories with starting points denoted by \Large$\bullet~$\footnotesize and end points by $\blacksquare$. The red circles are the location estimates at each iteration. The figure is best viewed in color.}
\end{figure}

 Figure~\ref{fig:scenario_10target} shows a challenging scenario involving 10 targets maneuvering while being detected by five sensors that limited fields of view (with overlaps). The targets enter and exit the scene mainly from the sides of the scene, with different birth/death times and locations. The birth process is simply chosen to have a single LMB component, hypothesizing that with a constant existence probability of 0.05, one object may appear at each time step with its location uniformly distributed within the surveillance area. The five sensors fields of view and a sample of 10 target trajectories are also shown in Fig.~\ref{fig:scenario_10target}. Each sensor has a $60^\circ$ wide field of view, except in studies investigating the effect of the width of sensor's field of view~(see section~\ref{FoV_analysis}).
\subsection{Probability of detection profile}
To model the detection profile of each sensor in our numerical studies, we adopt a model similar to the one presented in~\cite{pDcomputation} for sonar detectors used for submarine location estimation. The detection profile of each sensor is not only range-dependent~(decreases with increasing sensor-target distance), but also angle-dependent and has no detection ability when the targets are outside the sensor's field of view.

Assuming a Rayleigh fading signal model, with a threshold $\Gamma_{th}$, the probability of detection is
\begin{equation}
p_{D,i}=\exp\Big(-\frac{\Gamma_{th}}{1+c_0f_1(r_i)f_2(\psi_i)}\Big),
\end{equation}
where  $c_0$ is a constant, $f_1(r_i)$ describes the dependence on the propagation distance between sensor and object, which has inversely proportional relationship $f_1(r_i)=r_i^{-1}$ in which $r_i(x) = \sqrt{(p_x-x_{s_i})^2+(p_y-y_{s_i})^2}$, and $f_2(\psi_i)$ specifies the dependence on the propagation angle. Bearing dependence function is  modeled similar to a Butterworth filter,
\begin{equation}
f_2(\psi_i)=\frac{1}{1+(\frac{\psi_i}{W})^{2H}},
\end{equation}
where 2W is the 3dB bandwidth corresponding to the sensor's field of view, and $H$ is the filter order, and $\psi_i$ is the $i$-th sensor-target vertical bearing given by
\begin{equation}
\psi_i(x) =\arctan\Big(\Big|\frac{p_x-x_{s_i}}{p_y-y_{s_i}}\Big|\Big).
\end{equation}
When a target is outside the FoV of the $i$-th sensor, the false alarm rate is 
\begin{equation}
p_{fa,i}=\exp(-\Gamma_{th}),
\end{equation}
which can be controlled by adjusting the threshold $\Gamma_{th}$.

\begin{figure}
\centering
\includegraphics[width=3 in]{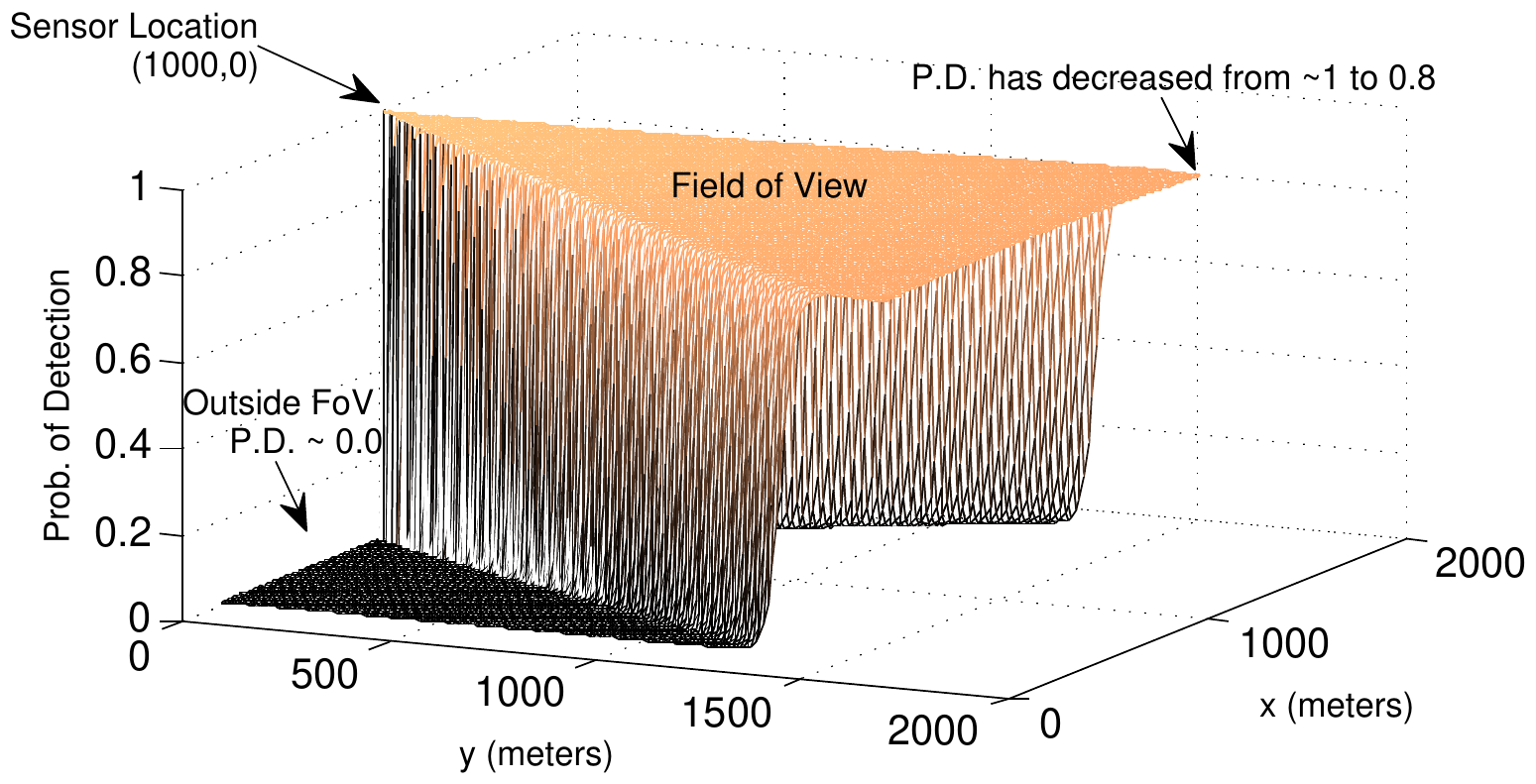}\\
\caption{\label{fig:pD} Sensor's detection probability distribution for $W=30^\circ$. }
\end{figure}

Choices of the above parameters depend on the specific application environment. In our case, we select a higher order value $H=40$ to ensure that sensors have a good detection ability especially for the boundary of each sensor's field of view. We also set $\Gamma_{th}=4.6$ and $c_0=4.6\times10^4$ to ensure each sensor's $P_{D,i}$ is between 0.99 and 0.01. Figure~\ref{fig:pD} shows one sensor's detection probability distribution for $W=30$. Applying this detection profile to sensors in the scenario shown in Fig.~\ref{fig:scenario_10target} leads to different probabilities of detection by different sensors for the same target. For instance, for a target located at (400,1000), we have $p_{D,s_1}=0.9221$, $p_{D,s_2}=0.9238$, $p_{D,s_3}=0.9156$, $p_{D,s_4}=0.01$ and $p_{D,s_5}=0.01$.

\subsection{Different clutter analysis}
Each measurement set acquired from each sensor also includes Poisson distributed clutter with varying clutter rates. In our first experiment presented here, we examined the tracking performance of our multi-sensor fusion method with different clutter rates including $\lambda_c = 5$, $\lambda_c = 15$ and $\lambda_c = 30$. We ran 200 Monte Carlo repetitions for the same object trajectories but with independently generated different clutter and measurement noise realizations. We then computed the average OSPA errors with order and cut-off parameters $p=2$ and $c=100$ (see~\cite[Eqs.~(3)-(4)]{vo_OSPA_TSP_2008} for definition of OSPA and its parameters).

Figures~\ref{fig:different_lambda} and~\ref{fig:card_different_lambda} show the OSPA errors and cardinality statistics for each case. The OSPA peaks at $k=10$, $k=15$, $k=30$ and $k=35$ are due to the objects' appearance and disappearance at those times. The results show that our proposed method exhibit desirable tracking performance that does not substantially degrade in severely high clutter rate environments. They also demonstrate that variances of cardinality statistics are small which is evidence for desirable reliability of detection capability of the proposed tracking solution.

Note that the state of art (Consensus fusion for labeled random set distribution using GCI~\cite{consensus_labelled_RFS_2016}) failed to track all the existing targets, because at any time there were targets that were not detected by one or more sensors, hence there was no consensus between all sensors about them and they were omitted from the fused results. 

\begin{figure}
\centering
\includegraphics[width=3 in]{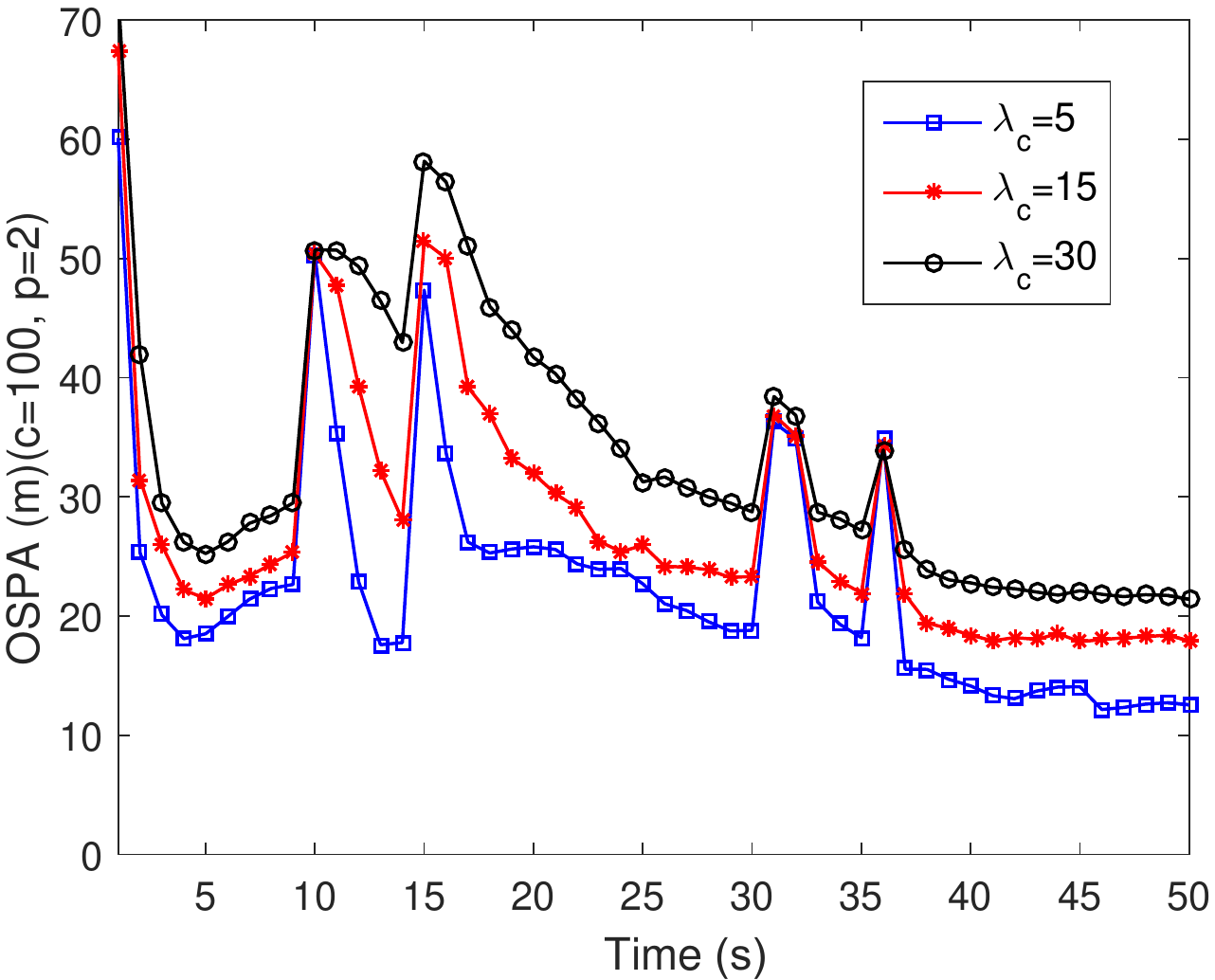}
	\caption{\label{fig:different_lambda} OSPA distances under low clutter rate $\lambda_c=5$, middle clutter rate $\lambda_c=15$ and high clutter rate $\lambda_c=30$ (averaged over 200 MC runs.) The figure is best viewed in color.}
\end{figure}
\begin{figure}
\centering
\includegraphics[width=3 in]{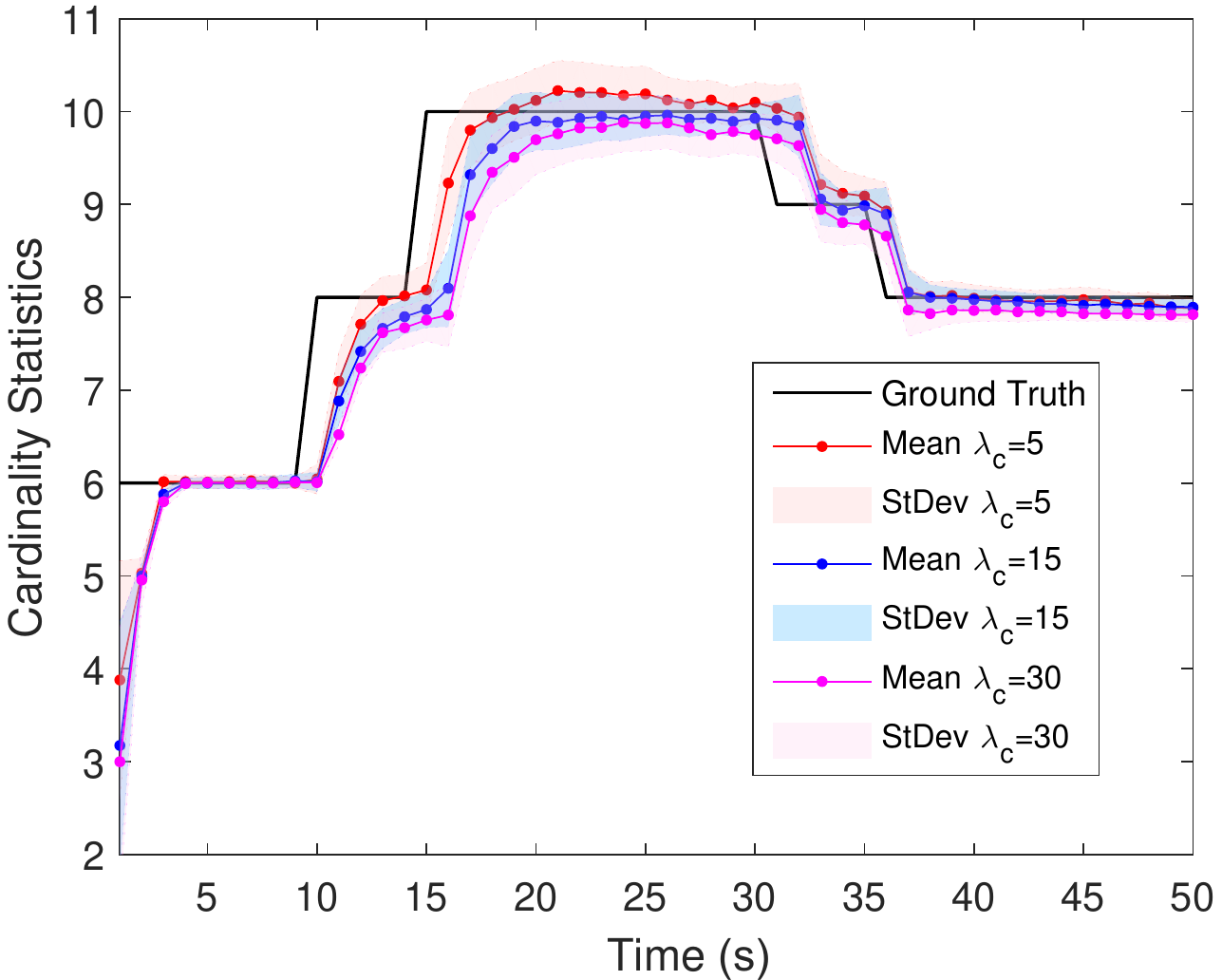}\\
\caption{\label{fig:card_different_lambda}Cardinality statistics under low clutter rate $\lambda_c=5$, middle clutter rate $\lambda_c=15$ and high clutter rate $\lambda_c=30$ (averaged over 200 MC runs.) The figure is best viewed in color.}
\end{figure}

\subsection{Comparative Analysis}
\label{FoV_analysis}
In order to compare the performance of our method with the state of art GCI fusion method as introduced in~\cite{consensus_labelled_RFS_2016} and used in~\cite{meng_CS_2016,bailu_2016}, we examined the OSPA errors and cardinality statistics returned by our proposed sensor fusion method with GCI fusion. To implement GCI fusion (with fixed and equal weights, as proposed in \cite{meng_CS_2016,bailu_2016,consensus_labelled_RFS_2016}), we simply used the bandwidth $W = 90^\circ$ and disabled the weight adaptation part of our code (replacing it with all weights equal to $1/n_s$). A bandwidth of $W = 90^\circ$ means unlimited sensors' field of view. Hence, all sensors should have consensus on detecting an existing target, and the GCI fusion method (which is consensus-based) should work. We also recorded the results returned by GCI fusion method in cases where the sensors' field of view is limited (to demonstrate that the GCI fusion method fails while ours performs satisfactorily). 

Figures~\ref{fig:different_width} and~\ref{fig:card_different_width} present the OSPA errors and cardinality statistics of each method in each scenario, averaged over 200 Monte Carlo experiments in each case. The following observations can be made from the results presented in Figures~\ref{fig:different_width} and~\ref{fig:card_different_width}:
\begin{itemize}
\item When the sensors have unlimited field of view ($W = 90^\circ$), our method and the state of art GCI fusion method~\cite{consensus_labelled_RFS_2016} perform similarly. But our method returns slightly lower estimation errors, because within its sensor fusion scheme, it prioritizes sensor measurements that convey more information.
\item When the field of view is limited (e.g. $W = 30^\circ$ in our experiments), the GCI fusion method~\cite{consensus_labelled_RFS_2016} fails to detect some targets (it returns cardinality estimates that are far below the ground truth in Figure~\ref{fig:card_different_width}) and consequently returns large OSPA errors--see Figure~\ref{fig:different_width}. This is while with the same limited field of view, our method returns acceptable cardinality estimates (detecting all targets in most of the 200 Monte Carlo runs) and small OSPA errors.
\item The performance of our proposed method (in terms of detection - cardinality - and overall OSPA errors) is robust to variations in the field of view. There is not much difference in data plotted for our method in figures~\ref{fig:different_width} and~\ref{fig:card_different_width} for different $W$ values.
\end{itemize} 

\begin{figure}
\centering
\includegraphics[width=3 in]{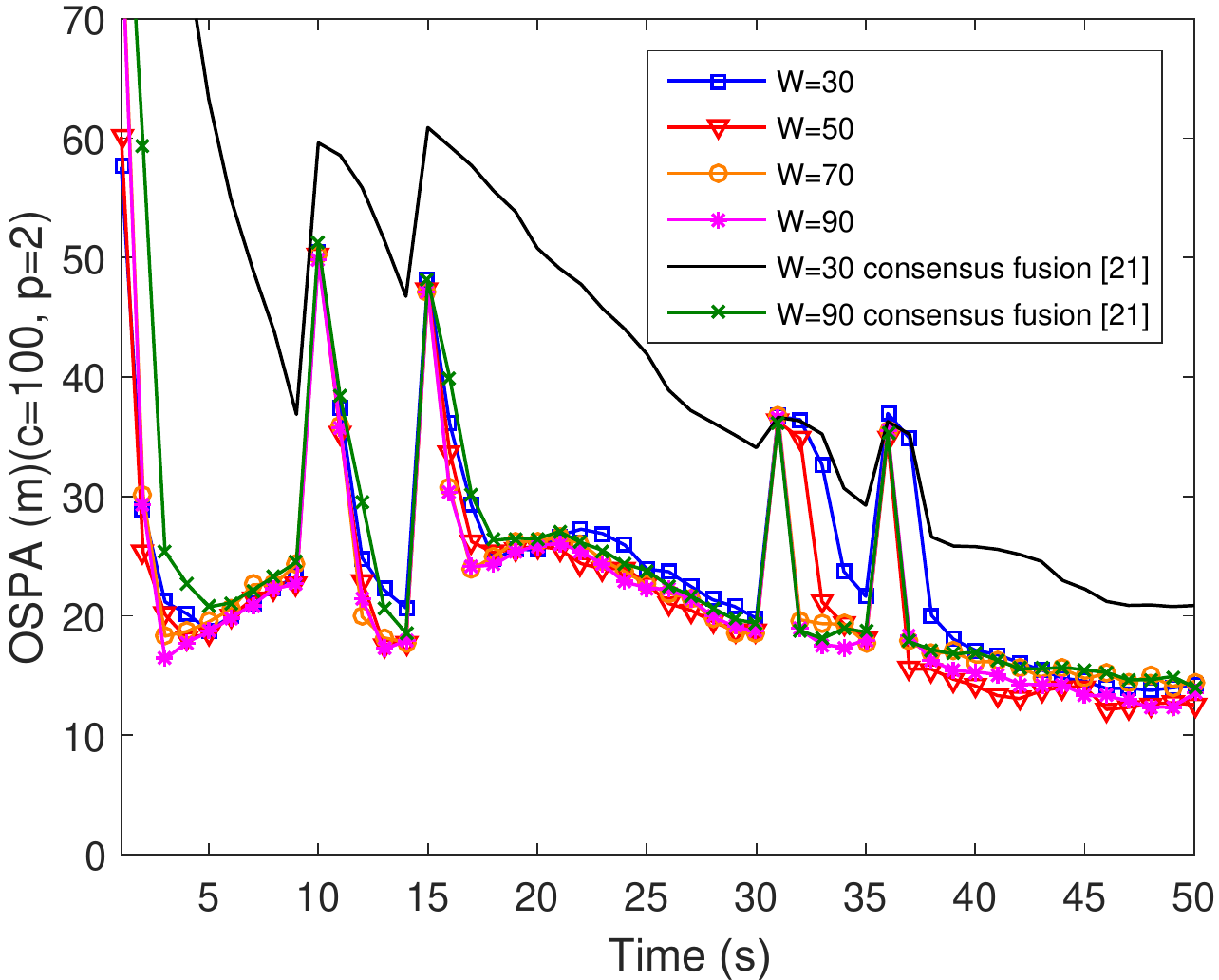}\\
\caption{\label{fig:different_width}OSPA distances with different FoV of sensors while $\lambda_c=5$ (averaged over 200 MC runs.). The figure is best viewed in color.}
\end{figure}

\begin{figure}
\centering
\includegraphics[width=3 in]{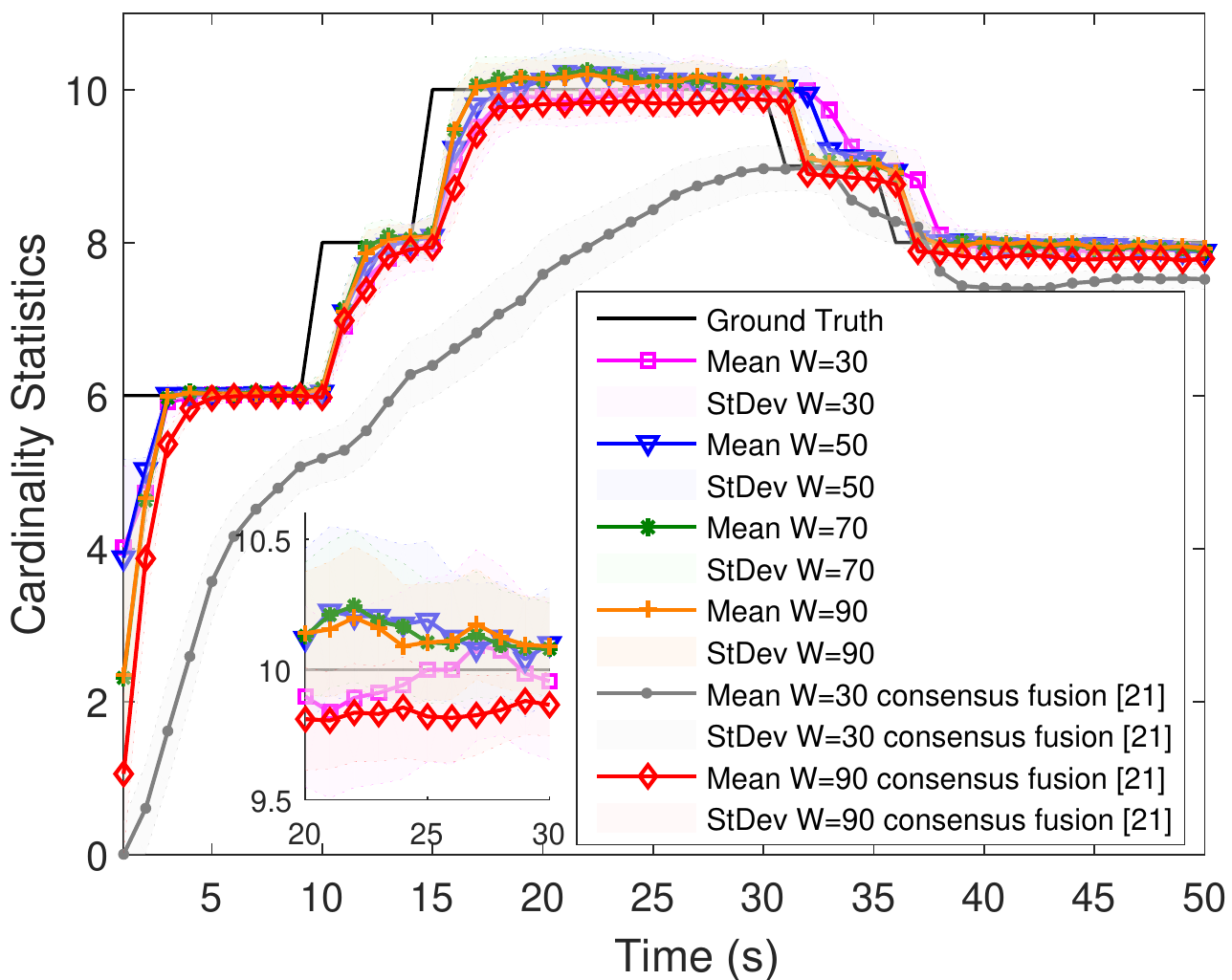}\\
\caption{\label{fig:card_different_width}Cardinality statistics with different FoV while $\lambda_c=5$ (averaged over 200 MC runs.). The figure is best viewed in color.}
\end{figure}

\section{Conclusions}
\label{sec:conc}
A novel multiple-view multi-sensor fusion method was introduced for LMB filters. Through the Cauchy-Schwarz divergence evaluation for each LMB component on all sensors, the proposed method overcomes the drawbacks of commonly used Generalized Covariance Intersection (GCI) method, which only considers a constant weight to each sensor during the whole tracking and fusion process. With a proper sensor importance weights selection scheme, a step-by-step SMC fusion implementation algorithm was detailed with LMB filters running on each sensor node. Numerical experiments involving a challenging multi-target tracking scenario, showed that our method can properly fuse information received from sensors with different fields of view, while the state of the art sensor fusion method fails.

\section*{Acknowledgment}
This project was supported by the Australian Research Council through ARC Discovery grant~DP160104662, as well as National Nature Science Foundation of China grant~61673075.

\end{document}